\documentclass[12pt,reqno]{amsart} 

\usepackage{amssymb} 
\usepackage{amsmath} 
\usepackage{fullpage} 
\usepackage{url}
\usepackage{color}

\newtheorem{theo}{Theorem}
\newtheorem{prop}{Proposition}
\newtheorem{lemma}{Lemma}

\theoremstyle{definition}
\newtheorem{defi}{Definition}
\theoremstyle{remark}
\newtheorem{rem}{Remark}
\newtheorem{exa}{Example}

\def\Rb{{\mathbf R}}

\def\Xb{{\mathbf X}}
\def\Yb{{\mathbf Y}}

\def\pib{{\boldsymbol\pi}}

\def\Ar{{\mathbb A}}

\def\Er{{\mathbb E}}

\def\Pr{{\mathbb P}}
\def\Qr{{\mathbb Q}}
\def\Rr{{\mathbb R}}

\def\Ac{{\mathcal{A}}}

\def\Cc{{\mathcal{C}}}

\def\Fc{{\mathcal{F}}}

\def\Nc{{\mathcal{N}}}

\def\Pc{{\mathcal{P}}}

\def\Sc{{\mathcal{S}}}

\def\Xc{{\mathcal{X}}}

\def\ba{{\rm \bf ba}}
\def\one{{\rm \bf 1}}
\def\({\left(}     
\def\){\right)}    
\def\[{\left[}     
\def\]{\right]}

\def\as{{\frenchspacing a.s.}~}
\numberwithin{equation}{section}

\begin{document}
\title{Fairness principles for insurance contracts in the presence of default risk} 
\author{Delia Coculescu and Freddy Delbaen}
\address{Institut f\"ur Banking und Finance, Universit\"at Z\"urich, Plattenstrasse
   14, 8032 Z\"{u}rich, Switzerland}
\address{Departement f\"ur Mathematik, ETH Z\"urich, R\"{a}mistrasse
   101, 8092 Z\"{u}rich, Switzerland}
   \address{ 
   Institut f\"ur Mathematik,
 Universit\"at Z\"urich, Winterthurerstrasse 190,
 8057 Z\"urich, Switzerland}
\date{\today}

\begin{abstract}
We use the theory of cooperative games for the design of fair insurance contracts. An insurance contract needs to specify the premium to be paid and a possible participation in the benefit (or surplus) of the company. It results from the analysis that when a contract is exposed to the default risk of the insurance company, ex-ante equilibrium considerations require a certain participation in the benefit of the company to be specified in the contracts.  The fair benefit participation of agents appears as an outcome of a game involving  the residual risks induced by the default possibility and using fuzzy coalitions. 
\end{abstract} 
\maketitle
\section{Introduction}

In insurance theory in contrast to financial theory, we cannot use the equilibrium considerations of the arbitrage pricing theory for valuation purposes. Such a use would require  that agents are able to sell short risk exposures or to hold risky positions of their choice. Instead,  in insurance, the principle of insurable interest holds and this imposes that agents  entering  insurance contracts have an exposure to the insured  risks. The only available actions by agents remain decisions on accepting or not accepting the proposed insurance contracts. In this situation, cooperative game theory is the right  tool, as it provides equilibrium concepts that can be used for valuation purposes. This paper is intended to exploit this idea by characterising fair contracts in presence of default risk. It is worth noting that similarly to the arbitrage theory, we will  avoid dealing with utility functions of economic agents and only focus on the relation between the premium paid and the payoffs of the contracts.

The common approach in insurance is to use so-called premium calculation principles, having plausible properties from an economic standpoint. One such largely agreed upon property is convexity, which ensures that diversification of risks, which is desirable from an economic pointy of view, is accounted for in the premium calculation. Fair allocation of total premia to individual contracts using coherent risk measures and cooperative game theory was already employed in the literature (see \cite{Pisa}); fair allocations being defined as allocations in the core of a cooperative game. We shall use this type of framework as a starting point.

But the approaches so far completely left out  the default risk of the insurance company from the analysis.  Indeed, in presence of default risk, the insured agents may remain exposed to a residual risk even after signing an insurance contract, as they might not fully receive the promised indemnities.  
Furthermore, the insured agents are not similarly exposed to the default risk of the insurance company, their loss in default being determined by the dependence between individual risk and total risk. Thus, the question of the fairness of the insurance contracts comes in naturally  whenever the insurance company may default. 

We choose to treat the problem of default risk as a time 1  allocation, or allocation of contracts' random payments, as opposed to the premia allocation, which is a time 0 allocation, or cost allocation. This time 1 allocation consists in characterising  the payoffs that are to be paid at time 1 for each contract, depending on the realisation of the risks and whether the company defaults or not.   In our approach, agents can form coalitions and compare costs and residual risks of the corresponding defaultable contracts. The residual risk of one agent measures the impact of the default risk of the insurance firm on the corresponding contract.

Problems related to allocations of random payoffs  have a long history in insurance, for instance in the study of the optimal reinsurance or optimal risk transfer. Already Borch \cite{Borch62} mentions cooperative game theory as a mean of selecting among Pareto equilibria in a reinsurance problem.  More recently and using coherent functions (either utilities or risk measures) there is a vast literature on optimal risk transfer that treats also time 1 allocations. We refer to Heath and Ku  \cite{HeaKu04},  Barrieu and El Karoui \cite{BarrElK05}, \cite{BarrElK05a}, Jouini et al. \cite{JouSchTou08},  Filipovi\'c and Kupper \cite{FilKup08}, Burgert and R\"uschendorf \cite{BurRus06}. Conceptually, the approach in the current paper is distinct from the literature on optimal risk transfers using coherent functions. Main differences are underlined in the title. First, we are interested in fair payoffs of the contracts, while the literature on risk transfers aims at characterising Pareto optima, given specified utilities for agents. 
We will introduce fairness at the same time as Pareto optimality. Instead of defining utility functions for the agents, we use a predefined, coherent price system and the desire of agents of paying the least of costs. Secondly, and more fundamentally,  in the current paper the main question is the default risk of the insurance contracts.  In the above mentioned literature, agents hold risky positions and at time 1 in equilibrium, some agents pay other agents (hence the name of risk transfer). This is not a classical situation in the context where agents are the insured: in general, insured agents do not make  ex-post monetary transfers to compensate for losses that other insured agents incurred. Consequently,  we propose that all payments at time 1 are made from and within the limits of the existing capital of the insurance firm. As claims in default are determined by bankruptcy procedures, we introduce the class of admissible payments that fulfil bankruptcy priority rules.

The remaining of this paper is organised as follows: in Section \ref{sec:setup} we introduce the precise mathematical setting of premium calculation principles that we are using, which are commonotonic submodular functions. Section \ref{sec:ecmodel} presents the insurance model the economic problem to solve. Section \ref{secFairness} formultates the fairness principles in a game theoretic framework. Solving the time 1 allocation requires additional theoretical development. This is introduced in Section \ref{SecAD} where we propose the concept of  state  coalitions  as a mean of reducing the set of solutions. This allows a more precise characterisation of fair payoffs and the proofs of the main results follow immediately. 

\section{Setup and notation}\label{sec:setup}
We work in a simple model consisting of two dates: time 0, where everything is known, and a fixed future date, time 1, where randomness is involved. For the purpose of representing the possible outcomes at time 1, a  probability space $(\Omega, \Fc, \Pr)$ is fixed.  Unless otherwise specified, all equalities and inequalities involving random variables are to be considered in an $\Pr$ a.s. sense.
 The space of  losses (or claims) occurring time 1 is considered to be $L^\infty(\Omega, \Fc, \Pr)$, simply denoted $L^\infty$, i.e.,  the collection of all essentially bounded random variables.

At time 0, the (manager of the) insurance company needs to evaluate the liabilities at time 1. The valuation method chosen will have an impact on the pricing of the individual insurance contracts as it will be explained later on.  We shall assume that the total liability in the company's balance sheet is evaluated by means of a convex functional $\Pc:L^\infty\to \Rr$ fulfilling the properties detailed below.
 \begin{defi}\label{defP} A mapping $\Pc\colon L^\infty\rightarrow \Rr$ is called a convex valuation function if the following properties hold
 \begin{enumerate}
 \item if $0\le \xi\in L^\infty$ then $\Pc(\xi)\ge 0$.
 \item $\Pc$ is convex i.e. for all $\xi,\eta\in L^\infty$, $0\le \lambda\le 1$ we have $\Pc(\lambda \xi +(1-\lambda)\eta)\le \lambda \Pc(\xi) + (1-\lambda) \Pc(\eta)$,
 \item for $a\in \Rr$ and $\xi\in L^\infty$, $\Pc(\xi +a)=\Pc(\xi) + a$
 \item if $\xi_n\uparrow\xi$ (with $\xi_n\in L^\infty$) then $\Pc(\xi_n)\rightarrow \Pc(\xi)$.
 \end{enumerate}
 If moreover for all $0\le \lambda\in \Rr$, $\Pc(\lambda \xi)=\lambda \Pc(\xi)$, we call $\Pc$ coherent.
 \end{defi}
 
 The number $\Pc(\xi)$ may be interpreted as a risk adjusted valuation of the future uncertain position $\xi$. More specifically,  we shall consider $\Pc$ to be a premium principle that the insurance company uses to determine the total premia to be collected from the insured, when the insurance company faces liability $\xi$ at time 1. This reflects the view  that insurance premia should be dependent on the whole portfolio of insurance contracts (see Deprez and Gerber \cite{DG}).
 
  Property (1) in the definition is therefore clear: liability tomorrow brings cash today for the insurance firm; or, conversely, premium collected time 0 corresponds to  liability in insurer's balance sheet. The convexity property is a translation of the diversification benefit. Combinations of risks are less risky than individual positions and, in presence of high competition on the insurance market, it is reasonable to believe that the benefits from the diversification will be passed along to the insured, at least partially, through lower premia.  This explains property (2).  Property (3) means that risk adjusted valuations are measured in money units.  Of course money time 0 is different from money at the end of the period. Introducing a deflator or discounting -- as is the practice in actuarial business since hundreds of years -- solves this problem.  It complicates notation and as long as there is only one currency involved it does not lead to confusion if one supposes that this discounting is already incorporated in the variables. The fourth property is a continuity property.  Using monotonicity (a consequence of the previous properties, see \cite{FDbook}), we can also require that $\Pc(\xi_n)\uparrow \Pc(\xi)$.  
 The homogeneity property is a strong property.  
 
 \begin{rem}
 A coherent valuation function $\Pc$ is a submodular function. If we put $\rho(\xi)=\Pc(-\xi)$ we obtain a coherent risk measure (having the Fatou property); alternatively $u(\xi)=-\Pc(-\xi)$ defines a coherent utility function (see \cite{Del} for more details). Submodular functions that are commonotonic are commonly used in insurance as nonlinear premium principles. We will come back to this  below, once we define commonotonicity.
 \end{rem} 
 
 We say that a random variable $\xi$ is acceptable if $\Pc(\xi)\le 0$.  Remark that $\xi-\Pc(\xi)$ is always acceptable.  If $\Pc$ is coherent, then the acceptability set 
$$
\Ac:=\{\xi\mid \Pc(\xi)\le 0\}
$$ 
is a convex cone.  In this paper, the term ``acceptability" refers to the insurance company and its balance sheet, and not to agents' preferences. Acceptable risks have positive value for the insurance company to hold, hence it would accept holding these risks without requiring a premium payment in exchange.  

The continuity assumption allows to apply convex duality theory and leads to the following representation theorem
\begin{theo} If $\Pc$ is coherent, there exists a convex closed set $\Sc\subset L^1$ (with $L^1$ being the space of   all equivalence classes of integrable random variables on $(\Omega, \Fc, \Pr)$), consisting of probability measures, absolutely continuous with respect to $\Pr$, such that for all $\xi \in L^\infty$:
\begin{equation}\label{reprPS} 
\Pc(\xi)=\sup_{\Qr\in\Sc} \Er_\Qr[\xi].
\end{equation}
Conversely each such a set $\Sc$ defines a coherent valuation function.
\end{theo}

We will consider such a set  $\Sc$ as given and fixed through the analysis, and we will refer to it as the ``scenario set'' generating $\Pc$. An assumption that we will make in this text is that $\Pc$ is coherent with the set $\Sc$ being weakly compact. The weak compactness property of $\Sc$  is equivalent to a  more stronger continuity property of $\Pc$ than in Definition \ref{defP} (4), where $\uparrow$'s are replaced by $\downarrow$'s.  The weak compactness ensures that given $\xi\in L^\infty$ one can find a probability measure  $\Qr^\xi\in \Sc$ such that 
$$
\Pc(\xi)=\Er_{\Qr^\xi} [\xi].
$$  Indeed,  the weak compactness of $\Sc$ is also equivalent to  the weak subgradient of $\Pc$ at $\xi$, $\nabla \Pc(\xi)$,  being non empty for all $\xi\in L^\infty$.  We refer to \cite{FDbook} for the precise statements and proofs of equivalent formulations for $\Sc$ being weakly compact.

In this text we will also make the assumption  that $\Pc$ is commonotonic.
\begin{defi}  We say that two random variables $\xi,\eta$ are commonotonic if there exist a random variable $\zeta$ as well as two non-decreasing functions $f,g\colon \Rr\rightarrow\Rr$ such that $\xi=f(\zeta)$ and $\eta=g(\zeta)$.
\end{defi}
It is easily seen that two random variables $\xi,\eta$ such that $\xi\le 0,\eta\ge 0$ and $\{\xi<0\}\cap\{\eta>0\}=\emptyset$, are always commonotonic.
\begin{defi} We say that $\Pc\colon L^\infty\rightarrow\Rr$ is commonotonic if for each couple $\xi,\eta$ of commonotonic random variables we have $\Pc(\xi+\eta)=\Pc(\xi)+\Pc(\eta)$.
\end{defi}
\begin{rem} 
 Commonotonic convex monetary valuation functions are positively homogeneous and hence coherent. 
  \end{rem}
  
Addition of commonotonic risks is the opposite of diversification. Indeed, $\xi$ and $\eta$ being nondecreasing functions of $\zeta$, neither of them is a hedge against the other. The commonotonicity of $\Pc$ can therefore be seen as a translation of the rule:  if there is no diversification, there is also no gain when putting these claims together. Including commotonicity as an economic principle is considered to be desirable for the purpose of premium calculation or risk measurement in insurance  and many actuarial models are built on the assumption that the premium principles are commonotonic. For instance, the popular class of concave distortion risk measures are commonotonic risk measures, hence they can be seen as particular examples fitting in the framework we develop below. 

With $\Pc$ commonotonic, one can prove the  following result (this will be useful later on):
\begin{lemma}\label{lemsim1}
Let $\xi \in L^\infty$ and $m \in \Rb$. Then,  the following hold:
\begin{align}\label{qstarZplus}
\nabla \Pc(\xi )\subset &\nabla \Pc((\xi-m)^+)\\\label{qstarminus}
\nabla \Pc(\xi )\subset &\nabla \Pc(-(\xi-m)^-)\\
\nabla \Pc(\xi )\subset& \nabla \Pc(\xi \wedge m).
\end{align}

\end{lemma}
\begin{proof} 
Using commonotonicity of $\Pc$ we find 
\begin{equation}\label{intm}
\Pc\(\xi-m\)=\Pc\((\xi-m)^+\)+\Pc\(-(\xi-m)^-\)
\end{equation} 
 Any  $\Qr^*\in \nabla \Pc(\xi)$ satisfies:
\begin{align*}
\Pc(\xi-m)&=\sup_{\Qr\in\Sc}\Er_{\Qr}[\xi-m]=  \Er_{\Qr^*}[\xi-m]\\
\Pc\((\xi-m)^+\)&\geq  \Er_{\Qr^*}\((\xi-m)^+\)\\
\Pc\(-(\xi-m)^-\)&\geq  \Er_{\Qr^*}\(-(\xi-m)^-\)
\end{align*}
 then replacing these expressions  in (\ref{intm}) we find that 
 \begin{align} \Pc\(\xi-m\)&=\Pc\((\xi-m)^+\)+\Pc\(-(\xi-m)^-\)\\ &\geq   \Er_{\Qr^*}[\xi-m]=\Pc(\xi-m),
 \end{align} 
 in other words we only have equalities. This proves that $\Qr^*\in \nabla \Pc((\xi-m)^+)$ and $\Qr^*\in \nabla \Pc(-(\xi-m)^-)$. Similarily, by commonotonicity and the above: $$\Pc(\xi \wedge m)=\Pc(\xi)-\Pc\((\xi-m)^+\)= \Er_{\Qr^*}[\xi\wedge m],$$ proving  $\Qr^*\in \nabla \Pc(\xi\wedge m)$.
\end{proof}

There is a link between commonotonic valuation functions and Choquet integration theory (see Schmeidler \cite{Schm}). For bounded nonnegative risks $\xi$,  the following representation holds when $\Pc$ is a commonotonic valuation function:
\begin{equation}\label{Palter}
\Pc(\xi)=\sup_{\mu\in\Cc(w)}\Er_\mu[\xi]=\int_0^\infty w(\xi>a)da
\end{equation}
where $w:\Fc\to\Rr_+$ satisfies $w(\Omega)=1$ and
\begin{equation}\label{2alter}
w(A\cap B)+w(A\cup B)\leq w(A)+w(B),
\end{equation}
and  
$$
\Cc(w):=\{\mu \text{ finitely additive measure}\;|\; \mu(\Omega)=1,\; \forall A\in\Fc:\;0\le \mu(A)\leq w(A)=\Pc(\one_A) \}.
$$
The use of Choquet integration as premium principle was emphasized by Denneberg, \cite{Denn1}. Denneberg was inspired by the pioneering work of Yaari, \cite{Yaa}.

A set function $w$ satisfying (\ref{2alter}) is called 2-alternating or supermodular and can serve as a characteristic function for cost games (which are the duals of convex games) and $\Cc(w)$ is called the core of the game.  In general $\Sc\subset \Cc(w)$, but when $\Sc$ is weakly compact we have $\Sc =\Cc(w)$ and the following representation of the scenario set $\Sc$ holds:
$$
\Sc=\{\mu \text { probability measure}\;|\; \mu\ll \Pr,\; \mu(A)\leq w(A), \forall A\in\Fc\}.
$$
In Delbaen \cite{D1} and Schmeidler \cite{Schm1} one can find the basics of convex game theory that are going to be used in this paper. 

\bigskip

To sum up, the properties of $\Pc$ that are going to be used in this paper are:\\
\noindent 
\textbf{Assumptions.} The  valuation function $\Pc$ is commonotonic and the scenario set $\Sc$  (cf. the representation in (\ref{reprPS})) is weakly compact in $L^1$.

\section{The economic model}\label{sec:ecmodel}

We consider $N$ economic agents ($N\in\mathbb N^*$) that are the potential buyers of insurance.  Each agent $i$ is endowed with a risky position $X_i\in L^\infty_+$  and is considering at time 0 the option to buy insurance for covering his/her exposure. The insurer will be denoted by the index $0$; his role is to propose insurance contracts to each agent. Also, he provides an initial capital denoted by $k_0$.  We shall  propose some principles for the design of fair individual insurance contracts. The only assumption about the agents preferences is that they prefer more rather than less and they are risk adverse, as they seek to buy insurance. Therefore, when comparing different contracts that are proposed to them, the premium to be paid is an important element in their decision making. They choose the lower premium when the indemnities are the same. We will have to deal with the issue that the default risk of the insurance company is depending on the premia collected, impacting the cash flows of the contracts at time 1. Therefore random cash flows need to be compared as well and not only premia. This will be detailed in a while (Section \ref{secFairness}). 
 We shall assume the insurer's contribution  $k_0$ exogenous to the analysis,  for instance it reflects some exogenous regulatory constraints and shareholder preferences regarding the riskiness of the company's shares. The question of setting a level of the capital $k_0$ is not addressed in this paper.
 
Therefore, we fix $\Xb=(X_0,X_1,...,X_N)$, i.e. an $N+1$ dimensional vector of random variables nonnegative and bounded, where $X_1,...,X_N$ are the risks to be insured, while $X_0=k_0$ by convention, i.e, the equity capital brought by the insurer at time 0. For simplicity, from now on we shall refer to $\Xb$ as the \textit{liability vector}, even though, the first component $X_0$ is strictly speaking not liability.  From now on,  $N+1$ dimensional random vectors will always appear in bold.

The insurance company uses a commonotonic valuation function $\Pc$ for calculating the total premia to be collected.  Given that the total risk exposure is
$$
 S^{\Xb}:=\sum_{i=1}^N X_i,
 $$  the total premium is computed as:
\begin{equation}\label{premiatotal}
k:=\Pc\(S^{\Xb}\)
\end{equation}
so that the \textit{total capital}\footnote{We use this term in the frame of this paper to mean total fund that is available in the insurance firm at time 1; it does not match an accounting definition of total capital. } of the insurance firm is
$$
K:=k+k_0,
$$
and we will suppose that suppose that $\Pr(S^X>k)>0$ so that the total premia is not sufficient to exclude the possibility of default and $K>0$. These assumptions exclude some trivial cases from the analysis. 

Hence, we assume the insurance company charges the insured the minimum amount that makes the aggregate net position acceptable. This reflects the discussion from the previous section, where we assumed that there are benefits from pooling risks and these benefits are passed along to the insured, at least partly. The valuation of the total premia using a convex functional as in (\ref{premiatotal}) is a necessary first step in order to pass along some diversification benefits to the insured. Remains to ensure that each agent individually  ``gets her share of the pie''; here the concept of fairness of the contracts comes into play. 

Let us more precisely describe the components of an insurance contract.

\subsection{Premia, indemnities and benefit sharing}
The insurance company is proposing insurance contracts to the $N$ agents. In particular, for the risks   $X_1,\cdots X_N $ and with an initial capital $X_0=k_0$, the insurance company proposes contracts 
 $$
 ((\pi_1,X_1,B_1),\cdots,(\pi_N, X_N,B_N))\in (\mathbb R_+\times L^\infty_+\times L^\infty_+)^N
 $$ i.e.,  a collection of individual contracts, with each individual contract $(\pi_i,X_i,B_i )$  specifying a premium $\pi_i\geq 0$ to be paid by the agent $i$ at time $0$ and a \textit{promised} payoff of the contract at time 1, consisting of the indemnity $X_i$ and possibly a participation in the benefit  of the company, $B_i\geq 0$ (dividend). One may want to consider  all $B_i=0$, i.e., there is no such benefit sharing proposed in the contracts. In this case, the insurer will receive the whole benefit in the form of a dividend. We shall deal with the issue of the benefit sharing using equilibrium considerations, rather than taking from the outset the point of view that there is or not such a participation proposed within a contract. 
 
The actual payments cannot exceed the total capital of the insurance company, and hence the insurance company defaults when $S^{\Xb}>K$.  The contracts are defaultable  and therefore the promised payment $X_i+B_i$ may differ from the \textit{actual} payment of the contract, denoted $Y_i$. This satisfies
\begin{align*}
Y_i&=X_i+B_i\text{ if the company does not default}\\
Y_i&<X_i \text{ if the company defaults on its contracts.}
\end{align*}
It follows that the actual payment and the benefit participation are linked by the relation
$$
B_i=(Y_i-X_i)^+.
$$
As the insurer brings in an initial capital $k_0$, he is entitled at time $1$ to a dividend  
 $$
 Y_0:=\(K-\sum_{i=1}^NY_i\)^+,
 $$
 that is, the payoff to the insurer is the residual value, once all payoffs to the insured are paid out.

All payments by the insurance company occurring at time 1 will be called \textit{payoffs}; they are indemnities and dividends and they represent the actual monetary transfers from the insurance firm toward the agents, as opposed to the promised payments of the contracts.  We shall restrict our analysis to the class of payoffs that are  feasible, i.e., they do not exceed the total capital available and also respect the legal requirement that indemnities have a higher priority of payment over dividends. These are called admissible payoffs.

 \begin{defi} We denote by $\Xc$ the space of $N+1$ dimensional random variables which are bounded. 
 We consider the liability vector $\Xb=(X_0,...,X_N)\in \Xc$ with all $X_i\geq 0$ and $X_0$  constant and consider $K$ another constant satisfying $X_0\leq K$.
 \begin{itemize} 
 \item[(a)] We call payoff of total mass $K$ any vector $\Yb=(Y_0,\cdots,Y_N)\in \Xc$ that satisfies $\sum_{i=0}^N Y_i= K$ and each $Y_i$ is $\sigma(X_i,S^{\Xb})$ measurable. The class of payoffs of total mass $K$ is denoted by $\Xc(\Xb,K)$. 
 \item[(b)] The class of admissible payoffs of total mass $K$,  corresponding to the liability $\Xb=(X_0,...,X_N)\in \Xc$ is defined as:
 $$
 \Ar^\Xb(K)=
 \left \{ \Yb\in \Xc(\Xb,K) \biggm |  \forall i\geq 1:
 \begin{array}{l}
  Y_i\one_{\{S^{\Xb}>K\}}=\frac{X_i}{S^{\Xb}} K\one_{\{S^{\Xb}>K\}}\\
  Y_i\one_{\{S^{\Xb}\leq K\}}\geq X_i\one_{\{S^{\Xb}\leq K\}}\\
 \end{array}
 \right  \}.
 $$
 \end{itemize}
 \end{defi}
 
Admissible payoffs respect the following rules:  all insurance claims $X_i$, $i\geq 1$ are paid entirely  to the insured whenever there is sufficient capital to do so ($S^{\Xb}\leq K$); default occurs when the total insurance claim exceeds the total capital ($S^{\Xb}>K$);  in case of default no dividend is distributed and the insured have equal priority of their claims. Hence in default all capital $K$  is distributed as indemnities, proportionally to the claim size. The fact that an admissible payoff component $Y_i$  is considered $\sigma(X_i,S^{\Xb})$ measurable ensures that we can specify ex-ante within the contracts the form of the payoffs as functions  of the individual risk and total risk only. This ensures that each agent can observe time 1 their payoff and judge if it respects the contract, without necessarily having knowledge of the losses of the other agents individually.  This assumption reflects the practice in the industry.  From a mathematical point of view it is not a necessary condition though. 
  
\begin{exa} {\it Standard payoffs}. Let us consider that each insured receives some constant proportion of the company's benefit $\(k-S^{\Xb}\)^+$. In this case, these constant proportions can be specified in the contracts at time $0$. The corresponding admissible payoffs  (that we shall call standard payoffs) are given as follows:
 \begin{align}\label{formYi}
 Y_i&=\[X_i+\alpha_i\(k-S^{\Xb}\)\]\one_{\{S^{\Xb}\leq k\}}+X_i\(\frac{K}{S^{\Xb}}\wedge 1\)\one_{\{S^{\Xb}> k\}},\quad i=1,\ldots,N\\\label{formY0}
 Y_0&=\[k_0 +\alpha_0\(k-S^{\Xb}\)\]\one_{\{S^{\Xb}\leq k\}}+\(K-S^{\Xb}\)^+\one_{\{S^{\Xb}> k\}},
\end{align}
where $k=K-X_0=K-k_0$ is the total premium and  each $\alpha_i$ is a nonnegative constant  and  $\sum_{i=0}^N\alpha_i=1$. The standard allocations have the feature that the surplus (or benefit)  $(K-S^{\Xb})^+$ is shared between the agents and the insurer as follows:
\begin{itemize}
\item[-]  when $S^{\Xb}\in[k,K]$, the insurer takes all the surplus, as it does not exceed his initial contribution (the equity capital) $k_0$.
\item[-] when $S^{\Xb}<k$, there is a benefit and all agents and the insurer receive some fixed share of the benefit.
\end{itemize}
For the insured the benefit writes:
$$
B_i=\alpha_i\(k-S^{\Xb}\)^+.
$$
 \end{exa}

 \begin{defi}
 The payoff vector $\Yb$ is called standard if they are of the form (\ref{formYi})-(\ref{formY0}). The corresponding contracts $(\pi_i,Y_i)$ are then called \textit{standard}.
 \end{defi}

\subsection{Problem description}

A liability vector $\Xb\in \Xc$ is fixed, with $X_0=k_0\geq 0$ constant. The design of a contract corresponding to the risk $X_i$ requires specifying two components: $\pi_i$, the premium  and $B_i\geq 0$, i.e., a possible benefit participation. This  is the same as specifying a premium $\pi_i$ and a random payment $Y_i$,  the actual contract payment. We saw above that the  actual payment $Y_i$ and the the benefit participation are linked by:
$$
B_i=(Y_i-X_i)\one_{\{S^{\Xb}<K\}}= (Y_i-X_i)^+.
$$
Therefore, the contracts' design problem reduces to the specification of a cost vector $\pib=(\pi_0,..\pi_N)\in\Rr_+^{N+1}$ satisfying $\pi_0=k_0$ and  $\sum_{i=0}^N\pi_i=K$, and a payoff vector $Y=(Y_0,...,Y_N)\in \Ar^{\Xb}(K)$. 

When are costs $\pib$ and payoffs $\Yb$ fair? We will take the point of view of each stakeholder (insurer or insured) and identify individual and collective rationality conditions for accepting the costs $\pib$ and payoffs $\Yb$  that are proposed to them. The stakeholders' objectives are partly cooperative and partly conflicting. 
On the one hand, by cooperating, i.e., pooling the risks, diversification is achieved and this is beneficial for the insured globally, as it reduces the total cost to be payed by a group. This total cost reduction  on the other hand needs to be allocated among all stakeholders and here objectives are conflicting: each one aims at paying a low cost $\pi_i$, while securing the largest payoff $Y_i$ for themselves . Cooperative game theory is  providing solutions concepts in this context. Two distinct problems need to be considered: 
\begin{itemize}
\item[-] Time 0 allocation, or cost allocation:  allocate the total capital $K$ to contributions of individual agents at time 1. The aim is to determine $\pib=(\pi_i)_{i=0}^N\in\Rr_+^{N+1}$, such that $\sum_{i=0}^N \pi_i =K$. 
For consistency with our initial assumptions, the allocation method will lead to $\pi_0=k_0=X_0$, that is, the shareholder will bring in the equity capital entirely.
\item[-] Time 1 allocation, or payoff allocation:  conditionally  on the realisation of the vector $\Xb$, allocate the total fund $K$ to the stakeholders. The aim is to select a random vector $\Yb\in\Ar^\Xb(K)$  with $Y_i$ being the actual payment to agent $i$.  
\end{itemize}

\section{Cooperative games and the notion of fairness of insurance contracts}\label{secFairness}

In this section we define the fair contracts as outcomes of cooperative games.  
The notion of a fair cost allocation that is based on equilibrium  in a cooperative game,  was introduced already in  Delbaen  \cite{Pisa}, \cite{FDbook}. We go further and also employ convex games for solving the time 1 allocation problem, that is,  the determination of the family of random payoffs $\Yb$.

  In  \cite{Pisa} it was shown that starting from any given risk vector $\Xb=(X_0,...,X_N)$ and given a coherent valuation function $\Pc$,  one can define a  cooperative game as follows.  Let $\Nc:=\{0,...,N\}$ be the set of players, consisting of all agents and the insurer; $2^{\Nc}$ is the set of all possible groups of players, named coalitions. We refer to $\Nc$ as the grand coalition.
We then define a \textit{cost game} $(\Nc,2^{\Nc},c^{\Xb})$, with  $c^{\Xb}$ being the characteristic cost function of the game,  $c^{\Xb}:2^{\Nc}\to\mathbb R$ defined as:
$$
c^{\Xb}(S):=\Pc\(\sum_{i\in S} X_i\),\text{ for all }S\in 2^\Nc.
$$  
In the literature, one often encounters a  value function of the game instead of a cost function. The value function $v^{\Xb}$ is obtained as  $v^{\Xb}(S):=c^{\Xb}(\Nc)-c^{\Xb}(\Nc\setminus S)$; $v^{\Xb}$ expresses value (or profit) and eplayers aim to get the most possible, while $c^{\Xb}$ measures costs that players intend to minimise. As we study cost allocation here, it is more intuitive to use the cost function $c^{\Xb}$, instead of the value function $v^{\Xb}$.  All concepts from profit games can easily be translated to the cost setting.

The set of actions available to a coalition $S$  consists of all possible divisions $(x_i ), i\in S$ of $c^{\Xb}(S)$ among the members of $S$, $\sum_{i\in S} x_i = c^{\Xb}(S)$. We search for an action of the grand coalition that is stable, in the sense that no coalition $S$ can obtain a better outcome. The set of such actions is called the core of the game (see Shapley \cite{Shapley} for more details):
$$
\Cc\(c^{\Xb}\):=\left \{\mathbf x=(x_i,i\in \Nc)\;\Big|\;  \sum_{i\in \Nc} x_i = c^{\Xb}(\Nc) \text{ and for all } S\in 2^\Nc \;  \sum_{i\in S} x_i \leq  c^{\Xb}(S)\right \}
$$
We follow \cite{Pisa} and propose the following definition:
\begin{defi}
An element  $\pib=(\pi_i,i\in \Nc)\in \Cc(c^{\Xb})$ is called  \textit{a fair cost allocation} for the risk vector $\Xb$. 
 \end{defi}
We remark that if $\pi\in\Cc\(c^{\Xb}\)$ we necessarily have that $\pi_0=k_0$.  Indeed $\pi_0\le c^{\Xb}(k_0)=k_0$ and $\sum_{j\ge 1}\pi_j\le \Pc(\sum_{j\ge 1}X_j)=k$.  But $\sum_{j\ge 0}\pi_j=K=k_0+k$ and hence the two inequalities must be equalities.
 
 The characteristic function $c^{\Xb}$ gives for any coalition the value of its liability, using the valuation function of the insurance company $\Pc$. From the standpoint of the agents in a coalition $S$,  $c^{\Xb}(S)$  represents a cost they have to bear for insurance, should they decide to split from the grand coalition.  A fair premia is therefore simply any  repartition of the total capital  $K$ in $N+1$ individual costs representing the contribution of each agent, such that no coalition prefers to split from the grand coalition in order to become a separate entity to be insured. 

Once  the costs allocated,  the next problem is determining the payoffs of the contracts, in particular a fair benefit allocation. This requires to specify a preference order on the space of payoff vectors $\Xc$, i.e.,  the space of bounded, $N+1$ dimensional random variables. As said previously, we do not assume any specific utility function for the agents, but simply that they are preferring to pay less rather than more and are risk adverse. In other words, they have increasing utility functions and care about the cost of insurance. 

The following definition introduces the most natural preference order for each coalition, namely the cheapest to insure payoff is preferred by any group of agents. Given that any  agent $i$  has already an exposure $X_i$, this gives the following:

\begin{defi}\label{def:preference}
Given two payoff vectors $\boldsymbol \xi,\boldsymbol \eta \in \Xc$, we say that the coalition $S\in 2^\Nc$ prefers $\boldsymbol \xi$ to $\boldsymbol \eta$, and we write $\boldsymbol \eta \prec_S\boldsymbol \xi$ if:
$$
\Pc\(\sum_{i\in S}(X_i-\xi_i)\)<\Pc\(\sum_{i\in S}(X_i-\eta_i)\).
$$
or, alternatively:
$$
c^{\Xb-\boldsymbol\xi}(S)<c^{\Xb- \boldsymbol\eta}(S).
$$
The weak preference relation is denoted by  $\preceq_S$, namely we write $\boldsymbol \eta \preceq_S\boldsymbol \xi$ if the coalition $S$ does not prefer $\boldsymbol \eta$ to $\boldsymbol \xi$; indifference is denoted by $\sim_S$.
\end{defi}

\begin{rem}
We notice that with our definition of preferences, a coalition $S$ aggregates value from each of its members and does not care about allocations of costs or payoffs  among its members:   whenever $\boldsymbol  \xi,\boldsymbol \eta \in \Xc$ satisfy $\sum_{i\in S}\xi_i=\sum_{i\in S}\eta_i$, the coalition is indifferent between the two vectors:
$$
\boldsymbol \xi\sim_S\boldsymbol \eta.
$$
This is nothing but the concept of transferable utility, a common assumption in cooperative games.

Below we are going to apply preferences to net payoffs, which are payoffs from insurance minus costs of insurance.  It is therefore important to notice that a coalition $S$ will be indifferent between $\mathbf 0=(0,...,0)$ and any payoff $\boldsymbol \xi$ satisfying $\sum_{i\in S}\xi_i=0$. 
\end{rem}

\begin{rem}
Any risk $\boldsymbol \xi\in\Xc$ satisfying $\sum_{i\in S}(X_i-\xi_i)\in \Ac$ satisfies  $\mathbf 0\preceq_S \boldsymbol \xi$.
\end{rem}

We propose the following
\begin{defi}\label{deffairpayoff}
 Given a fair cost allocation $\pib$, a  \textit{fair payoff} for the risk vector $\Xb$, is a vector of random variables $\Yb\in\Ar^{\Xb}(K)$ satisfying 
 \begin{equation}\label{fairpayoff}
\mathbf 0\preceq_S \Yb-\pib \quad \forall S\in2^\Nc,
 \end{equation}
 with $\mathbf 0=(0,...,0)\in\Xc$.
 \end{defi}
 As a member of the grand coalition, player $i\in \Nc$ pays $\pi_i$ and receives at time 1 the  payoff $Y_i$. Therefore $\Yb-\pib$ represents the vector of net payoffs of players. A coalition $S$ has therefore a net payoff $\sum_{i\in S} (Y_i-\pi_i)$. Alternatively,  a coalition $S$ has a total cost $c^{\Xb}(S)$ to be paid at time 0 and a total payoff of $c^{\Xb}(S)$ that can be distributed to members of the coalition at time 1. The net payoff at the coalition level is therefore 0.  The interpretation of (\ref{fairpayoff}) is therefore the following: fair payoffs  are such that any coalition $S$  prefers $\Yb-\pib$ to a null net cash flow, that is, prefers buying insurance within the grand coalition rather than doing this as a separate entity.  

An alternative --- maybe easier to understand --- expression for  (\ref{fairpayoff}) is the following 
\begin{equation}\label{fairpayoffbis}
c^{\Xb-\Yb}(S)+\sum_{i\in S}\pi_i\leq c^{\Xb}(S)\quad \forall S\in2^\Nc,
\end{equation} where $c^{\Xb-\Yb}$ is the value of the residual claims of the coalition $S$, 
$$
c^{\Xb-\Yb}(S)=\Pc\(\sum_{i\in S} (X_i-Y_i)\).
$$ 

The intuition behind the expression (\ref{fairpayoffbis}) is as follows. As usual, player $i$ is exposed to risk $X_i$; we consider each player aims at achieving  full  coverage of their risk at the lowest possible cost. To achieve full coverage of risks, agents consider also the cost to insure residual risks $\Pc(X_i-Y_i)$, they need to pay additional insurance cost if  this quantity is positive. Every coalition $S$ looks at the proposed cost allocation $\pib$ and payoff allocation $\Yb$ within the grand coalition  and compare with the stand alone situation. 
 
Hence, when the payoff $Y_i$ is distributed to player $i$ as a member of the grand coalition, it leaves agent $i$  with a residual risk $X_i-Y_i$. All players now need to cover their residual risks. They can decide to form coalitions for this purpose. The residual cost of  achieving an acceptable position for  a coalition $S$ is  $ \Pc\(\sum_{i\in S} (X_i-Y_i)\)=c^{\Xb-\Yb}(S)$. Hence, all players may consider entering a second game (the re-insurance game), with characteristic function $c^{\Xb-\Yb}$. 

Alternatively, a coalition $S$, by paying a premia $c^{\Xb}(S)$ achieves directly an acceptable residual value and reinsurance is costless. Indeed,
$$
\sum_{i\in S}X_i -c^{\Xb}(S)\in \Ac.
$$
If the cost allocation and the payoffs of  the initial game $(\Nc,2^{\Nc},c^{\Xb})$ are fair,  any coalition can achieve an acceptable position at a  lower cost as members of the grand coalition and with possibly a reinsurance of the residual risk, rather than forming from the outset a separate group.

It is important to emphasize that in general the residual risk $X_i-Y_i$ of an agent is non zero, because of the default risk of the company. 

For this reason, in addition to fairness, it seems important that payoffs are tailored in a way that achieves Pareto optimality,  in the cost minimisation problem of the residual risks of all agents:
\begin{defi}
A payoff family $\Yb\in\Ar^{\Xb}(K)$ is said to be maximal if it is a solution of the following cost minimisation problem:
\begin{equation*}
\inf_{\boldsymbol \xi \in\Ar^{\Xb}(K)} \sum_{i=0}^N\Pc(X_i-\xi_i).
\end{equation*}
\end{defi}
\begin{lemma}\label{lemmaximal}
Standard payoffs are maximal.
\end{lemma}
\begin{proof}
For any $\boldsymbol \xi, \boldsymbol \eta\in \Ar^{\Xb}(K)$ we have 
\begin{align}\label{in}
 \sum_{i=0}^N\Pc((X_i-\xi_i)^+) &=\sum_{i=0}^N\Pc((X_i-\eta_i)^+).
\end{align}
Indeed, we have that the acceptability of $\boldsymbol \xi$ implies for all $i\in\Nc$:  $\{X_i-\xi_i>0\}\subset \{S^\Xb>k\}$ and  the payoffs in default are entirely fixed by the admissibility conditions: $\xi_i\one_{\{S^\Xb>k\}}=\eta_i\one_{\{S^\Xb>k\}}$. It follows that $(X_i-\xi_i)^+=(X_i-\eta_i)^+$, $\forall i\in\Nc$, hence (\ref{in}) is verified.

For any random variable $\xi$, $\Pc(X_i-\xi)= \Pc\((X_i-\xi)^+\)+\Pc\(-(X_i-\xi)^-\)$ by commonotonicity of $\Pc$. Using this property and (\ref{in}), it follows that $\Yb$ is maximal  if and only if it is a solution of
\begin{equation}\label{iin}
\inf_{\boldsymbol \xi \in\Ar^{\Xb}(K)} \sum_{i=0}^N\Pc\(-(X_i-\xi_i)^-\).
\end{equation}
As $\Pc$ is convex, for any  $\boldsymbol \xi \in\Ar^{\Xb}(K)$: $ \sum_{i=0}^N\Pc\(-(X_i-\xi_i)^-\)\geq \Pc\(- \sum_{i=0}^N(X_i-\xi_i)^-\)= \Pc\(-(S^\Xb-K)^-\)$. So,  the quantity in (\ref{iin}) is bounded below by  $\Pc\(-(S^\Xb-K)^-\)$. 

  To conclude, we notice that  $\Yb$ standard, implies all $(X_i-Y_i)^-$ commonotonic, so that  $ \sum_{i=0}^N\Pc\(-(X_i-Y_i)^-\)=\Pc\(- \sum_{i=0}^N(X_i-Y_i)^-\)= \Pc\(-(S^\Xb-K)^-\)$. \end{proof}

We now return to the insurance problem and give some definitions for this framework.
\begin{defi}
 Consider the risks $(X_i)_{i=1}^N$ and some corresponding contracts $\{(\pi_i,Y_i)\}_{i=1}^N$ with initial capital $X_0=k_0$.
  We denote $\pi_0:=K-\sum_{i=1}^N\pi_i$ and $Y_0=K-\sum_{i=1}^NY_i$ .
  \begin{enumerate}
  \item The premia $(\pi_i)_{i=1}^N$  are said to be fair if $\pi_0=k_0$ and $\pib=(\pi_0,...,\pi_N)$ is a fair cost allocation for the risk vector $\Xb=(X_0,...,X_N)$.  
  \item   The contracts are  said to be fair if the premia $(\pi_i)_{i=1}^N$ are fair and, given $\pib=(\pi_0,\cdots, \pi_N)$, $\Yb=(Y_0,\cdots, Y_N)$ is a fair payoff.
    \end{enumerate}
\end{defi}

\section{Determining fair  contracts}\label{SecAD}

The aim of this section is to characterise some fair insurance contracts in presence of default risk. The main difficulty is the design of fair payoffs for the defaultable contracts. As a technique to tackle this problem, we propose in Subsection \ref{nononerlap}  a distinct cooperative game, where players are states of $\Omega$ and coalitions are elements of $\Fc$.  The interpretation is that states are competing to get the highest possible payoff under the threat to split from the grand coalition $\Omega$. Another interpretation  is that for each element $A\in \Fc$,  one can create an available contract  that can be entered at time 0 to cover the risk $\one_A$. That is,  Arrow-Debreu risks can be covered by insurance contracts\footnote{Such a market could be labelled as ``complete''; however this could  lead to some confusions, as we do not assume that individual agents can eliminate all risks, the contracts that they enter being defaultable.  Also, note that we do not assume that for an $A\in \Fc$ there is necessarily an agent $i\in\{1,..,N\}$ having a risk exposure of $\one_A$.  }
 and an agent can enter such a contract and even several such contracts, provided it does not exceed the agent's true risk exposure (by the principle of insurable interest). We will show that the fuzzy version of this game is in fact a generalisation of the game that we introduced in the previous section.  In Subsection \ref{sec:proofs}, we will prove that ensuring fairness in the ``Arrow-Debreu fuzzy game'' leads to fairness of the contracts in the frame of the initial  $N+1$ player game.

\subsection{Arrow-Debreu claims and the families of fair state payoffs}\label{nononerlap}
We denote 
$$
Z:=S^\Xb+k_0.
$$
Let us consider the following \textit{cost game} $(\Omega, \Fc, c^Z)$ with the characteristic function $c^Z:\Fc\to \Rr$ is defined as
$$
c^Z(A):=\Pc(Z\one_A),\quad A\in\Fc.
$$
Any element in $\Fc$ is now interpreted as a coalition, and $\Fc$ as the set of all coalitions. The cost $c^Z(A)=\Pc(Z\one_A)$ corresponds to the liability level of the insurance company,  relative to the set $A$.  By the principle of insurable interest, a coalition could not exceed this level of  exposure, that is the cumulative exposure of all agents.  But lower levels should be possible to reach by coalitions but they are currently excluded in this setting. 
For this reason, we introduce fuzzy coalitions, where  players $\omega\in\Omega$ ``choose" a rate of participation in a coalition instead of a binary decision yes/no:
\begin{defi}
 A fuzzy coalition is  a random variable  $\lambda:(\Omega,\Fc)\to [0,1]$. For  any fuzzy coalition $\lambda$, the corresponding cost function is $c^Z:[0,1]\to \Rb$ given as $c^Z(\lambda):=\Pc(Z\lambda)$.
 \end{defi} 

From an economical standpoint, an insurance company that commits to liability $Z$ implicitly  commits to any lower liability $0\leq \xi\leq Z$ as well. This level of risk corresponds to a fuzzy coalition $\lambda=\frac{\xi}{Z}$.  The fuzzy game 
$(\Omega, \Fc, c^Z)$ is a generalisation of the  finite-player game $(\Nc, 2^\Nc, c^\Xb)$ from the previous section. Indeed, the set of players $\Nc=\{0,...,N\}$ can be linked to  the following set of fuzzy coalitions $\{\hat \lambda_0,...,\hat \lambda_N\}$ with $\hat \lambda_i=\frac{X_i}{Z}$ so that the cost function of the finite-player game is obtained via the relation 
$$
c^\Xb(S)=c^Z\(\sum_{i\in S}\hat \lambda_i\)\text{, for any $S\in\Nc$. }
$$
 
\begin{prop}\label{convexcz} The set function $c^Z$ is 2-alternating, that is, it satisfies for all $A,B\in\Fc$:
$$
c^Z(A\cap B)+c^Z(A\cup B)\leq c^Z(A)+c^Z(B).
$$
This implies that the cost game $(\Omega, \Fc, c^Z)$ is superadditive.
\end{prop}
\begin{proof}For any $\xi_1,\xi_2\in L_+^\infty$, it holds that $c^Z(\xi_1\vee\xi_2)+c^Z(\xi_1\wedge\xi_2)\leq c^Z(\xi_1)+c^Z(\xi_2)$. Indeed, using the representation of $\Pc$ in (\ref{Palter})-(\ref{2alter}):
\begin{align*}
c^Z(\xi_1\vee\xi_2)&+c^Z(\xi_1\wedge\xi_2)=\Pc(Z(\xi_1\vee\xi_2))+\Pc(Z(\xi_1\wedge\xi_2))\\
&=\Pc((Z\xi_1)\vee(Z\xi_2))+\Pc((Z\xi_1)\wedge (Z\xi_2))\\
&=\int_0^\infty w\((Z\xi_1)\vee(Z\xi_2)>a\)da+\int_0^\infty w\((Z\xi_1)\wedge(Z\xi_2)>a\)da\\
&=\int_0^\infty w\(\{Z\xi_1>a\}\cup \{Z\xi_2>a\}\)da+\int_0^\infty w\(\{Z\xi_1>a\}\cap \{Z\xi_2>a\}\)da\\
&\leq \int_0^\infty w\(Z\xi_1>a\)da+\int_0^\infty w\(Z\xi_2>a\)da\\
&\leq \Pc\(Z\xi_1\)+\Pc\(Z\xi_2\)=c^Z(\xi_1)+c^Z(\xi_2).
\end{align*}
By taking $\xi_1=\one_A$ and  $\xi_1=\one_B$, we find the claimed inequality.
\end{proof}

In this framework, we consider a family:
\begin{equation}\label{contracts}
\{(\pi(A),Y(A)) ,A\in \Fc\},
\end{equation}
where for every set $A\in\Fc$,  $\pi(A)\in \Rb_+$   is interpreted as a cost and $Y(A)$ is a random variable interpreted as a payoff, both corresponding to the risk exposure $Z\one_A$.  It remains to  define the notions of fair cost and fair payoffs in this setting.
The definitions below are a generalisation of the ones with finite number of coalitions $2^{N+1}$ in Section \ref{secFairness}.  

\begin{defi}
\begin{itemize}
\item[(i)] A fair cost allocation is  any element $\pi$ in the core $\Cc(c^Z)$ of the game,
$$ 
\Cc(c^Z)= \left \{\nu\in \ba(\Omega,\Fc,\Pr)\;|\;\nu(\Omega)=K=c^Z(\Omega)\text{ and } \nu(A)\leq c^Z(A),\;\forall A\in \Fc\right \}.
$$
\item[(ii)]The fuzzy core $\widetilde \Cc(c^Z)$ of the game is
$$ 
\widetilde \Cc(c^Z)= \left \{\nu\in \ba(\Omega,\Fc,\Pr)\;|\;\nu(\Omega)=K=c^Z(\Omega)\text{ and } \int \lambda d\nu\leq \Pc\(\lambda Z\),\;\forall \lambda:(\Omega,\Fc)\to [0,1]\right\}.
$$
A fair cost allocation in the sense of fuzzy games  is  any element $\pi$ in the fuzzy core of the game.
\end{itemize}
\end{defi}
The notation $\ba(\Omega,\Fc,\Pr)$ above stands for the space of  bounded, finitely additive measures absolutely continuous with respect to $\Pr$. It is easily seen that $\widetilde \Cc(c^Z)\subset \Cc(c^Z)$.  Also the weak compactness of $\Sc$ implies that for $A_n\downarrow \emptyset$, we have $c^Z(A_n)\downarrow 0$.  The core $\Cc(c^Z)$ is therefore a set of sigma-additive measures, in other words $\widetilde \Cc(c^Z)\subset \Cc(c^Z)\subset L^1_+$ and the two cores are weakly compact convex sets.

The core is the set of allocations which cannot be improved upon by any coalition.  We can ``shrink'' it by introducing the fuzzy core, that is the set of allocations which cannot be improved upon by any fuzzy coalition.
 The basic papers for this approach are Aubin \cite{Au}, Artzner and Ostroy \cite{AO} and Billera and Heath \cite{BH}.  
 As explained above, using the fuzzy game approach is interesting as it leads to a generalisation of the game in the previous section.  
 
Another  advantage of using the fuzzy game  approach is that the fuzzy core can be characterised and linked to some elements in the scenario set $\Sc$:
\begin{prop}\label{fuzzycore} The following hold:
\begin{align*}
\widetilde \Cc\(c^Z\)&=Z\cdot\nabla \Pc(Z)\subset \Cc\(c^Z\)\\
Z\cdot \Sc&\subset \Cc\(c^Z\)-L^1_+
\end{align*}
If $Z>0$ $\as$ then
\begin{equation*}
\widetilde \Cc\(c^Z\) = \Cc\(c^Z\)\text{ implies } \Pc(Z)=\max_{\Qr\in\Sc}\Er_\Qr[Z]=\min_{\Qr\in\Sc}\Er_\Qr[Z].
\end{equation*}
\end{prop} 
\begin{proof}

We first show $\widetilde \Cc\(c^Z\)=Z\cdot\nabla \Pc(Z)$. By definition $\nabla \Pc(Z)=\Sc\cap\{\nu\mid \nu(\Omega)=K\}$. Take now $\nu\in Z\cdot \Sc$, then obviously $\nu(\xi)\le c^Z(\xi)$ for all $\xi\in L^\infty$. This shows that $Z\cdot\nabla \Pc(Z)\subset \widetilde \Cc\(c^Z\)$. Conversely, take $\nu\in \widetilde \Cc\(c^Z\)$ then  $\nu$ satisfies $\nu(\lambda)\leq \Pc(\lambda Z)$ for all $\lambda: (\Omega,\Fc)\to[0,1]$. This implies that $\nu$ satisfies $\nu(\xi)\leq \Pc(\xi Z)$ for all $\xi\in L^\infty$. Hence $\nu\in Z\cdot\Sc$. But then $\nu(\Omega)=K$ shows that $\nu\in Z\cdot\nabla \Pc(Z)$.

Let us now  prove that $Z\cdot \Sc\subset \Cc\(c^Z\)-L^1_+$. Take $\Qr\in\Sc$ and suppose that $\Qr\notin \Cc\(c^Z\)-L^1_+$. The set $\Cc\(c^Z\)$ is weakly compact and hence $ \Cc\(c^Z\)-L^1_+$ is a closed convex subset of $L^1$. The Hahn-Banach theorem then allows to find $\xi\in L^\infty$ such that
$$\Er_\Qr[\xi Z]>\sup\left\{\Er\[\xi\(\frac{d\nu}{d\Pr}-h\)\]\mid \nu \in \Cc\(c^Z\);h\in L^1_+\right\}.$$
This implies that $\xi\ge 0$ and hence $\sup\{\Er[\xi(d\nu/d\Pr-h))\mid \nu \in \Cc\(c^Z\);h\in L^1_+\}=\sup\{\nu(\xi)\mid \nu \in \Cc\(c^Z\)\}$. We get that
$$
\int_0^\infty \Er_\Qr[Z\one_{\{\xi>u\}}]\,du=\Er_\Qr[\xi Z]>\sup\{\nu(\xi)\mid \nu \in \Cc\(c^Z\)\}=\int_0^\infty c^Z(\xi>u)\,du,
$$
which is a contradiction since $\Er_\Qr[Z\one_{\{\xi>u\}}]\le c^Z(\xi>u)$ for all $u$ and all $\Qr\in\Sc$.
The inclusion can be read as follows. For all $\Qr\in \Sc$ there exists $\nu\in  \Cc\(c^Z\)$ so that $Z\cdot \frac{d\Qr}{d\Pr}\leq \frac{d\nu}{d\Pr}$.
 
If $\widetilde \Cc\(c^Z\)= \Cc\(c^Z\)$, then, for all  $\Qr\in \Sc$ there exists $\Qr_0\in  \nabla \Pc(Z)$ so that $Z\cdot \frac{d\Qr}{d\Pr}\leq Z\cdot \frac{d\Qr_0}{d\Pr}$.  If $Z>0\;\as$ then  $\frac{d\Qr}{d\Pr}\leq  \frac{d\Qr_0}{d\Pr}$ $\as$, that is, $\Qr=\Qr_0$ and hence $\Sc=\nabla \Pc(Z)$. 
\end{proof}
\begin{rem} The above proposition shows that the $\wedge,\vee$ inequality in the beginning of the proof of Proposition \ref{convexcz} is not equivalent to commonotonicity.  Indeed, the functional $c^Z$ is not commonotonic since it is different from $\sup\{\nu(\xi)\mid \nu \in \Cc\(c^Z\)\}$ which is the uniquely defined  commonotonic extension of $c^Z$ restricted to $\Fc$.
\end{rem}
We now define random payoffs $\{Y(A),A\in\Fc\}$ as introduced in  (\ref{contracts}) and study their fairness.
\begin{defi}\label{OptStateAlloc}A family of random variables $\{Y(A),A\in\Fc\}$ is called a \textit{family of state payoffs}  if $Y$ is a finitely additive vector measure on $\Omega$
\begin{align*}
Y: \Fc&\to L^\infty_+\\
A&\mapsto Y(A),
\end{align*}
that satisfies:
\begin{enumerate}
\item $Y(\Omega)=K$ \as.
\item For each $A\in \Fc$, $Y(A)$ is a random variable on $(\Omega,\sigma(Z, \one_A))$.
\end{enumerate}
$Y(A)$ shall be referred to as \textit{payoff for the coalition }$A$. 
\end{defi}
We make the choice that  payoffs $\{Y(A),A\in\Fc\}$ are measurable with respect to $\sigma(Z, \one_A)$ rather than $\Fc$. This is in order to have contracts that specify payoffs contingent on the risk realisation and do not include some other extraneous randomness. Here also, this condition is not needed from a mathematical point of view, but it is more realistic  from an economic point of view.  In the definition we only required that $Y$ is finitely additive.  But the interesting payoffs will become countably additive in the following sense.  If $(A_n)_n$ is a sequence of pairwise disjoint sets taken in $\Fc$, then
$$
Y(A_1)+Y(A_2)+...=Y(\cup_nA_n)=K.
$$
where the sum converges in probability (and not necessarily in $L^\infty$ norm).

We now define the admissible payoffs as those state payoffs that satisfy the following constraints\footnote{The existence of the equity  shifts the bankruptcy set from $\{Z>K\}$ (the default event in absence of equity, that is, when $S^{\Xb}=Z$) to $Z>K+k_0$ (that is, the default event with equity $k_0$). The existence of equity also introduces an asymmetry between players and the payments in default. We choose to ignore these aspects for now and fully focus on the distribution of the surplus. For this, we only model the surplus state $\{Z<K\}$ and its complement.  In the next subsection we take care of the precise effect of bankruptcy and emphasize  the positive effect of the bankruptcy procedures.  }:
\begin{defi}
A family of state payoffs $Y$ is admissible if $Y\in\Ar^{Z,\Fc}(K)$, where:
\begin{equation*}
\Ar^{Z,\Fc}(K)=\left \{Y \text{ state payoff }  \biggm | 
\begin{array}{l}
  A\subset\{Z\geq K\}\Rightarrow Y(A)\one_A=K\one_A\\
  A\subset\{Z< K\}\Rightarrow Y(A)\geq Z\one_A\\
 \end{array}
 \right  \}.
\end{equation*}
\end{defi}
\begin{lemma}\label{remgenad}Any state payoff $ Y=\{Y(A),A\in\Fc\}\in \Ar^{Z,\Fc}(K)$  has the representation:
\begin{align}\label{reprpay}
 Y(A)&=  \alpha(A)(K-Z)^++  (Z\wedge K)\one_{A},
\end{align}
where $\alpha(A)$ is an $\Fc$ measurable random variable with values in $[0,1]$. The mapping $\alpha:\Fc\to L^\infty_+$ is a finitely additive vector measure, normalised to 1, i.e. $\alpha(\Omega)=1$ \as.     We will refer to $ \alpha$ as the benefit sharing measure corresponding to $ Y$. 
\end{lemma}

\begin{proof} Take $Y\in \Ar^{Z,\Fc}(K)$. We claim that if $A\subset\{Z\geq K\}$ then $Y(A)\one_{\{Z\geq K\}}=K\one_A$ and if $B\subset\{Z < K\}$ then $Y(B)\one_{\{Z\geq K\}}=0$. Indeed, if $A\subset\{Z\geq K\}$, denoting $\widetilde A:=\{Z\geq K\}\setminus A$ and using the additivity, we observe that $Y(\{Z\geq K\})\one_{\{Z\geq K\}}=K\one_{\{Z\geq K\}}=\(Y(A)+Y(\widetilde A)\)\one_{\{Z\geq K\}}$. As $Y(A)\one_A=K\one_A$ and $Y(\widetilde A)\one_{\widetilde A}=K\one_{\widetilde A}$ by the definition of admissible payoffs, it follows that $Y(A)\one_{\widetilde A}=0$. If  $B\subset\{Z< K\}$ then we denote $\widetilde B:= \{Z\geq K\}\cup B$ and we obtain $Y(\widetilde B)\one_{\{Z\geq K\}}=\(Y(B)+Y(\{Z\geq K\})\)\one_{\{Z\geq K\}}\leq K\one_{\{Z\geq K\}}$. As $Y(\{Z\geq K\})\one_{\{Z\geq K\}}=K\one_{\{Z\geq K\}}$ and $Y$ is nonnegative, it follows that $Y(B)\one_{\{Z\geq K\}}=0$.  Hence both claims are proved. We obtain for $ A\subset\{Z\geq K\}$ and $B\subset\{Z< K\}$ the representations: 
\begin{align*}
Y(A) &= K\one_{A}+Y(A)\one_{\{Z< K\}}\\
Y(B)&= Z\one_{B}+\eta(B)\one_{\{Z< K\}}. 
\end{align*}
for some nonnegative random variable $\eta(B)$.

Starting from any set $E\in\Fc$ we can create a partition $A_1,A_2,B_1,B_2$ of $\Omega$ with:  $A_1=E\cap \{Z\geq  K\}$, $B_1=E\cap \{Z< K\}$, $A_2=E^c\cap \{Z\geq  K\}$ and $B_2=E^c\cap \{Z< K\}$. We have that  $Y(A_1)+Y(A_2)+Y(B_1)+Y(B_2)=K$ and we can use the decompositions of these payoffs determined above to deduce that $\(Y(A_1)+Y(A_2)+\eta(B_1)+\eta(B_2)\)\one_{\{Z< K\}}=(K-Z)\one_{\{Z< K\}}$. This ensures  the  existence of nonnegative random variables $\alpha(A_1), \alpha(A_2), \alpha(B_1),\alpha(B_2)$ satisfying $\sum_{i=1}^2 (\alpha(A_i)+\alpha(B_i))=1$ and such that . 
\begin{align*}
Y(A_i) &= K\one_{A_i}+\alpha(A_i)(K-Z)\one_{\{Z< K\}}\\
Y(B_i)&= Z\one_{B_i}+\alpha(B_i)(K-Z)\one_{\{Z< K\}}. 
\end{align*}
As $E=A_1\cup B_1$ with $A_1\cap B_1=\emptyset$, we have that $Y(E)=Y(A_1)+Y(B_1)$ and we obtain the representation (\ref{reprpay}) for $Y(E)$, with $\alpha(E)=\alpha(A_1)+\alpha(B_1)$. The finite additivity property of $\alpha$ is inherited from the one of $Y$. 
\end{proof}

\begin{defi} 
\begin{itemize}
\item[(a)] Given $\pi\in \Cc\(c^Z\)$, we say that a family of state payoffs $Y=\{Y(A),A\in\Fc\}$ is \textit{fair}, if
 for any $A\in \Fc$, the following inequality holds
\begin{equation}\label{optcond}
\Pc(Z\one_{A}-Y(A))+\pi(A)\leq \Pc(Z\one_A).
\end{equation}

\item[(b)] Given  $\pi\in\widetilde \Cc\(c^Z\)$, we say that a family of state payoffs $Y=\{Y(A),A\in\Fc\}$ is \textit{fair in the sense of fuzzy games}, if
 for any $\lambda:(\Omega,\Fc)\to  [0,1]$, the following inequality holds
\begin{equation}\label{optcond}
\Pc(\lambda Z-Y(\lambda))+\pi(\lambda)\leq \Pc(\lambda Z),
\end{equation}
where $\pi(\lambda)=\int \lambda (\omega)\pi(d\omega)$ and 
$$
Y(\lambda)=\int \lambda (\omega)Y(d\omega)
$$ is the integral\footnote{ The reader who is not familiar with  integration with respect to a finitely additive vector measure can easily find out that the definition for elementary functions $\lambda$ can be extended by using the $L^\infty$ density.} of $\lambda$ using the vector measure $A\mapsto Y(A)$.
\end{itemize}

\end{defi}

\begin{rem}\label{rempref}We can still interpret fairness of payoffs using preference orders on the space of bounded random variables. Indeed, Definition \ref{def:preference} can be  adapted to the setting where  coalitions are elements of $\Fc$.  We introduce for each coalition a preference order on $L^\infty(\Omega,\Fc,\Pr)$:  coalition $A\in \Fc$ prefers $\xi$ to $\eta$, and we write $\eta \prec_A\xi$ if and only if:
$$
\Pc\(Z\one_A-\xi \)<\Pc\(Z\one_A-\eta\).
$$
We denote weak preference  by  $\preceq_A$, and  indifference by $\sim_A$.
Then, $Y$ is a family of fair payoffs if  
$$\pi(A)\preceq_A  Y(A)\quad\forall A\in\Fc.$$
\end{rem}
We will direct our attention to the case of  those fair state payoffs that are also Pareto optimal for all possible finite partitions of $\Omega$. These are called maximal payoffs and are defined as follows:

\begin{defi}\label{defMaxY}  A family of  state payoffs $Y=\{Y(A),A\in\Fc\}$ is called  maximal if for any finite sequence of mutually disjoint sets $(A_n)$ satisfying $A_n\in \Fc$  and $\cup A_n=\Omega$ we have that $\{Y(A_n)\}_{n\geq 1}$ is a solution of the following cost minimisation problem:
\begin{equation}\label{maxY}
 \inf_{Y'\in\Ar^{Z,\Fc}(K)}\sum_n\Pc\(Z\one_{A_n}-Y'(A_n)\).
\end{equation}
\end{defi}
Maximal payoffs  are such that the costs of residual risks cannot be reduced even by reallocations of the risks in different portfolios (or coalitions), where a reallocation of risks corresponds to a partition $(A_n)$ of $\Omega$.

The following result gathers some properties of maximal payoffs:

\begin{theo}\label{theomax} 
\begin{itemize}
\item[(a)] Let $\alpha:(\Omega,\Fc)\to [0,1]$ be a finitely additive probability measure.  Then, the family of  state payoffs $Y$ associated to $\alpha$ via the relations (\ref{reprpay}) 
is maximal.

\item[(b)] If $\widetilde Y$ is  maximal, there is a finitely additive probability measure $\alpha:(\Omega,\Fc)\to [0,1]$ so that 
$$
 Y(A)\sim_A \widetilde Y(A),\quad \forall A\in\Fc,
$$
where the payoff $Y$ is  associated to $\alpha$ via (\ref{reprpay}).
\item[(c)] If $\widetilde Y$ is  maximal, then for   $\Qr^*\in\nabla \Pc(Z)$:
\begin{equation*}
\Pc\(-\(Z\one_A-\widetilde Y(A)\)^-\)=-\Er_{\Qr^*}\[\(Z\one_A-\widetilde Y(A)\)^-\], \forall A\in\Fc
\end{equation*}

\end{itemize}
\end{theo}
\begin{proof}
\begin{itemize}
\item[(a)] By Lemma \ref{remgenad}, a state payoff $\widetilde Y=\{\widetilde Y(A),A\in\Fc\}\in \Ar^{Z,\Fc}(K)$ writes:
\begin{align*}
\widetilde Y(A)&= \widetilde \alpha(A)(K-Z)^++  (Z\wedge K)\one_{A},
\end{align*}
where $\widetilde \alpha(A)$ is a  vector measure, i.e., the benefit sharing measure.  Using commonotonicity of $\Pc$ we obtain:
\begin{align*}
\Pc(Z\one_A-\widetilde Y(A))&=\Pc\(- \widetilde \alpha(A)(K-Z)^++(Z-K)^+\one_{A}\)\\
&=\Pc\(- \widetilde \alpha(A)(K-Z)^+\)+\Pc\((Z-K)^+\one_{A}\)
\end{align*}
and therefore the infimum in (\ref{maxY}) is obtained by a family of state payoffs $Y$ with its corresponding benefit sharing measure $ \alpha$ being a solution of
\begin{align}\label{interp}
 \inf_{\widetilde \alpha}\sum_k\Pc\(- \widetilde \alpha(A_k)(K-Z)^+\),
\end{align}
(the infimum being taken over the class of all vector measures on $\Fc$), and this for any  $(A_k)$, finite partition of $\Omega$. We observe that 
\begin{equation*}
\sum_k\Pc\(- \widetilde \alpha(A_k)(K-Z)^+\)\geq \Pc\(-  \sum_k\widetilde \alpha(A_k)(K-Z)^+\)= \Pc\(- (K-Z)^+\).
\end{equation*}
In case $(\widetilde\alpha(A_1),\widetilde\alpha(A_2),\cdots,\widetilde\alpha(A_n))$ is non random, we can use  homogeneity of $\Pc$ to show that we get equality above instead of an inequality, so that the infimum is attained by this vector. Therefore any finitely additive probability measure $\alpha$ on $\Fc$ is a solution of (\ref{interp}).  Furthermore the infimum in (\ref{maxY}) equals
$$
\Pc\(- (K-Z)^+\) + \sum_k \Pc\((Z-K)^+\one_{A_k}\).
$$

\item[(b)] In general, for a solution $\widetilde \alpha$ of (\ref{interp}), we denote
$$
\alpha(A):=\frac{\Pc\(- \widetilde \alpha(A)(K-Z)^+\)}{ \Pc\(- (K-Z)^+\)},\;\forall A\in\Fc.
$$
Then, $\alpha$ is a  finitely additive measure. Indeed it is nonnegative,  $\alpha(\Omega)=1$ and   finite  additivity is also verified:  for any $(A_k)$, a finite partition of $\Omega$, the random variables 
$$
\widetilde \xi_k:=- \widetilde \alpha(A_k)(K-Z)^+
$$
 satisfy --- as $\widetilde \alpha$ is a solution of  (\ref{interp}) --- 
\begin{equation}\label{intmax}
 \sum_k\Pc\(\widetilde\xi_k\)=\Pc\( \sum_k\widetilde\xi_k\)= \Pc\(- (K-Z)^+\).
 \end{equation}
 We see that the property in (b) holds  when $\alpha$ is as above: $Y(A)\sim_A \widetilde Y(A),\;\forall A\in\Fc$, as one can easily verify.

 \item[(c)] Let $\widetilde  \alpha$ be the benefit sharing measure associated with $\widetilde  Y$, so that we have for all $A\in \Fc$: $(Z\one_A-\widetilde Y(A))^-=\widetilde  \alpha(A)(K-Z)^+$.  Above, we  have seen that whenever $\widetilde  Y$ maximal, the equality (\ref{intmax}) holds, $\forall (A_k)k\text{ finite partition of }\Omega$.
From Lemma \ref{lemsim1}, for any  $\Qr^*\in\nabla \Pc(Z)$: 
$$
 \Pc\(- (K-Z)^+\)= \Er_{\Qr^*}\[- (K-Z)^+\]=- \sum_k\Er_{\Qr^*}\[\widetilde  \alpha(A_k)(K-Z)^+\].
$$ Using these equalities, we obtain
 $$
\sum_n\(\Pc\(- \alpha(A_n)(K-Z)^+\)+\Er_{\Qr^*}\[\alpha(A_n)(K-Z)^+\]\)= 0.
$$ which proves the statement, as the sum only contains nonnegative terms.
\end{itemize}
\end{proof}

\begin{theo}\label{reprfairstate} 
Let  $\Qr^*\in \nabla \Pc(Z)$ and $\pi^*:=Z\cdot \Qr^*\in\widetilde \Cc\(c^Z\)$.
 We consider state payoffs $Y^*=\{Y^*(A),A\in\Fc\}$:
\begin{align}\label{YA}
Y^*(A)&= \alpha^*(A)(K-Z)^++ (Z\wedge K)\one_{A},
\end{align}with $\alpha^*\ll \Pr$ satisfying 
\begin{equation}\label{alphaA}
\alpha^*=\frac{(Z-K)^+}{\Er_{\Qr^*}\[(Z-K)^+\]} \cdot \Qr^*.
\end{equation}
Then the following hold:
\begin{itemize}
\item[(a)] Given $\pi^*$, the family  $Y^*$ is fair. 
\item[(b)] 

Assume $\widetilde Y(A),A\in\Fc$ is  a family of state payoffs that is maximal and fair given $\pi^*$, in the sense of fuzzy games. Then there exists $\widehat \Qr\in\nabla \Pc(Z\wedge K)$ so that 
$$
\Pc(Z\one_A-\widetilde Y(A))= \Pc(Z\one_A-\widehat Y(A))\;\forall A\in\Fc,
$$
that is, (using preference relations in Remark \ref{rempref})
$$
\forall A\in\Fc\quad \widetilde Y(A)\sim_A \widehat Y(A),
$$
where $\widehat Y$ is defined as
\begin{align}\label{YAhat}
 \widehat Y(A)&=  \widehat \alpha(A)(K-Z)^++ (Z\wedge K)\one_{A},
\end{align}with
\begin{equation}\label{alphaAhat}
 \widehat \alpha(A):=\frac{\pi^*(A)-\Er_{\widehat \Qr}[(Z\wedge K)\one_{A}] }{\Pc\((Z-K)^+\)}.
\end{equation}
\end{itemize}
\end{theo}

\begin{proof}[Proof of Theorem \ref{reprfairstate}]
As $\alpha^*$ is a probability measure,  from Theorem \ref{theomax} (a), the family  $Y^*$ is admissible and maximal. 

Using  (\ref{qstarZplus}) in (\ref{alphaA}) we obtain that 
\begin{equation}
\alpha^*(A):=\frac{\Er_{\Qr^*}[(Z-K)^+\one_{A}]}{\Pc\((Z-K)^+\)}.
\end{equation}
By the commonotonicity and positive homogeneity of $\Pc$ and the expressions in Lemma \ref{lemsim1}:
 \begin{align*}
\Pc(Z\one_{A}-Y^*(A))&=\Pc(-(Z\one_{A}-Y^*(A))^-)+\Pc((Z\one_{A}-Y^*(A))^+) \\
&=\alpha^*(A)\Pc(-(K-Z)^+)+\Pc((Z-K)^+\one_{A})\\
&=-\alpha^*(A)\Er_{\Qr^*}[(Z-K)^+]+\Pc((Z-K)^+\one_{A})
\end{align*}
and
 \begin{align*}
\Pc(Z\one_A)&=\Pc\((Z\wedge K)\one_A\)+\Pc\((Z- K)^+\one_A\).
\end{align*}

Therefore, the inequality defining  fairness
$$
\Pc(Z\one_{A}-Y^*(A))+\pi^*(A)\leq \Pc(Z\one_A),\quad \forall A\in\Fc,
$$
writes:
\begin{equation}\label{interma}
\Er_{\Qr^*}[Z\one_A] - \alpha^*(A)\Er_{\Qr^*}[ (K-Z)^+]\leq \Pc\[(Z\wedge K)\one_A\],\quad \forall A\in\Fc
\end{equation}
that, after replacing $\alpha^*(A)$ with its expression, is equivalent to:
\begin{align*}
\Er_{\Qr^*}\[(Z\wedge K)\one_A\]\leq  \Pc\[(Z\wedge K)\one_A\] ,\quad \forall A\in\Fc.
\end{align*}
This is always satisfied, so that (a) is proved.

The claim in (b) is proved as follows.  First, one can check that $\widehat Y$ is fair given $\pi^*$, by replacing in the inequality (\ref{interma}) $\alpha^*$ with $\widehat \alpha$. 

Let  $\widetilde{Y}$ be maximal and fair. From Theorem \ref{theomax}  (b), it is sufficient to consider the case where its corresponding benefit sharing measure  $\widetilde \alpha$ is a finitely additive measure.  In this case, it is not difficult to check that $\widetilde Y(A)\sim_A \widehat Y(A)$ if and only if $\widetilde \alpha(A)=\widehat \alpha(A)$.

Therefore, it is necessary and sufficient to show that whenever $\widetilde Y$ satisfies: (i)  is fair in the sense of fuzzy games, (ii)  is maximal and (iii) its corresponding benefit sharing measure  $\widetilde \alpha(A)$  is a  finitely additive measure,  then there is $\widehat \Qr\in\nabla \Pc(Z\wedge K)$ so that $\widetilde \alpha(A)=\widehat \alpha(A)$.

Following similar steps as in the proof of (a) we find that the inequality 
\begin{equation}\label{intermf}
\Pc(\lambda Z-\widetilde  Y(\lambda))+\pi(\lambda)\leq \Pc(\lambda Z),\quad \forall \lambda \in L^\infty,\;0\leq \lambda\leq 1
\end{equation}
that characterises fair payoffs in the fuzzy game sense, is equivalent to
$$
\Er_{\Qr^*}[\lambda Z -\widetilde \alpha(\lambda) (K-Z)^+]\leq \Pc\[\lambda (Z\wedge K)\] ,\quad \forall \lambda \in L^\infty,\;0\leq \lambda\leq 1.
$$
We define the family $(\delta(A))$ with 
$$\delta(A):=\Er_{\Qr^*}[Z\one_A -\widetilde \alpha(A) (K-Z)^+].$$
 It can be verified that this is a finitely additive measure of total mass  $\Pc\(Z\wedge K\)$ and the inequality (\ref{intermf}) is equivalent to:
\begin{equation*}
\int\lambda d \delta\leq \Pc\[\lambda (Z\wedge K)\] ,\quad \forall \lambda \in L^\infty,\;0\leq \lambda\leq 1.
\end{equation*}

 In other words, $\delta$ is an element of the set
$$
 \left\{\nu\in \ba(\Omega,\Fc,\Pr) \;|\;\nu(\Omega)= \Pc\(Z\wedge K\), \int\lambda d \nu\leq \Pc\[\lambda (Z\wedge K)\] , \forall \lambda \in L^\infty\right \}.
$$
We obtain that there exists  $\widehat \Qr\in\nabla \Pc(X\wedge K)$, such that $\delta(A)= \Er_{\widehat \Qr}[(Z\wedge K)\one_A]$. Implicitly,  we find an expression for $\widetilde \alpha$:
\begin{align*}
\widetilde \alpha(A)&=\frac{\Er_{\Qr^*}[Z\one_A]- \Er_{\widehat \Qr}[(Z\wedge K)\one_A]}{\Er_{\Qr^*}[ (K-Z)^+]},
\end{align*}
that is exactly $\widehat \alpha$.
\end{proof}

\subsection{Fair contracts in the $N+1$ players game}\label{sec:proofs}
We consider the framework of Section \ref{secFairness}.
The next theorem gives the form of some fair insurance contracts. 
\begin{theo}\label{propFairStandard}
Suppose that $\Qr^*\in\nabla \Pc(S^\Xb)$ and 
$$
\pi_i=\Er_{\Qr^*}[X_i], \text{ for }i\in\{1,..,N\}. 
$$ Consider the   contracts  $\{(\pi_i,Y_i)\}_{i=1}^N$ with standard payoffs (i.e., the payoffs are as in (\ref{formYi})) such that the constant proportions $(\alpha_i)$ given by:
\begin{align}\label{alphai}
\alpha_i&=
\frac{\Er_{\Qr^*}\[ \frac{X_i}{S^{\Xb}} (S^{\Xb}-K)^+\]}{\Pc\[ (S^{\Xb}-k)^+\]}\text{ for }i\in\{1,..,N\}.
\end{align}
These contracts are fair and maximal. Furthermore, for the shareholders:
$$
\Pc(X_0-Y_0)=0
$$
and 
$$
k_0=\Er_{\Qr^*}[Y_0].
$$
\end{theo}

\begin{rem}We can associate with $S^{\Xb}$ a probability measure $\alpha^*:(\Omega, \Fc)\to[0,1]$:
\begin{equation*}
\alpha^*:=\frac{(S^{\Xb}-k)^+}{\Er_{\Qr^*}\[ (S^{\Xb}-k)^+\]}\;d\Qr^*
\end{equation*}
Then, one can verify that 
$$\alpha_i=\Er_{\alpha^*}\[\frac{X_i-Y_i}{\sum_{i=0}^N(X_i-Y_i)}\];\;i=0,1,...,N.$$
As the probability $\alpha^*$ assigns the entire mass to the default event, we can interpret the fair benefit share of agent $i$, $\alpha_i$  as being the expected value under $\alpha^*$ of the proportion of loss in default of agent $i$ over the total loss in default of all agents. 
\end{rem}

\begin{rem}
There is no uniqueness of the fair contracts.  The fair premia in Theorem \ref{propFairStandard} are known as an allocation in the fuzzy core of the game $c^\Xb$. Even if we consider a specific premium vector $\pi$ that is fixed, there is no uniqueness of the fair payoffs given $\boldsymbol \pi$. In Theorem \ref{reprfairstate}, the general forms of fair payoffs  given $\boldsymbol \pi$ was derived for the fuzzy game $c^Z$.
\end{rem}

\begin{proof}

The fact that $\boldsymbol \pi\in\Cc(c^\Xb)$, where $\pi_i=\Er_{\Qr^*}[X_i]$ for all $i\in\Nc$, is already a known result and trivial to verify directly. 

 By Lemma \ref{lemmaximal}, $\Yb$ are maximal. 
We now prove that $\Yb$ are fair payoffs given $\boldsymbol \pi$. As in the previous subsection, we use the notation
$$Z=S^{\Xb}+k_0
$$ 
and $$\lambda_i=\frac{X_i}{Z}\text { or } X_i=\lambda_iZ,\quad \text{ for }i=0,...,N.$$
 We use the notation from Subsection \ref{nononerlap}, in particular we take $Y^*$ to be the state payoff family defined  in Theorem \ref{reprfairstate}.  We have
 \begin{align*}
Y^*(\lambda_i)&=\int \lambda_i(\omega')Y^*(d\omega')=\Er_{\alpha^*}[\lambda_i](K-Z)^++\lambda_i(Z\wedge K)\\
&=\frac{\Er_{\Qr^*}[(Z-K)^+\lambda_i]}{\Pc\((Z-K)^+\)}(K-Z)^++(Z\wedge K)\lambda_i. 
\end{align*} 
By Theorem \ref{reprfairstate}, $Y^*=\{Y^*(A),A\in\Fc\}\in \Ar^{Z,\Fc}(K)$  is a family of state payoffs that are fair in the sense of fuzzy games, that is:
$$
\Pc(\lambda Z-Y^*(\lambda))+\pi^*(\lambda)\leq \Pc(\lambda Z),\quad\forall \lambda \in L^\infty,\;0\leq \lambda\leq 1.
$$
We recall that above $\pi^*=Z\cdot \Qr^*$ (see Theorem \ref{reprfairstate}). We can take in the above inequality $\lambda=\lambda(S)$, where $\lambda(S):=\sum_{i\in S}\lambda_i=\sum_{i\in S}\frac{X_i}{Z}$ for $S\subset \Nc$. We obtain that
\begin{equation*}
\Pc\(\sum_{i\in S}(X_i -Y^*(\lambda_i))\)+\Er_{\Qr^*}\[\sum_{i\in S}X_i\] \leq \Pc\(\sum_{i\in S}X_i\),\quad \forall S\in\Nc.
\end{equation*}
The inequalities above write:
\begin{equation}\label{intermineq}
\Pc\(\sum_{i\in S}(X_i -Y^*(\lambda_i))\)+\sum_{i\in S}\pi_i \leq c^\Xb(S),\quad \forall S\in\Nc.
\end{equation}
We now treat separately the case with and without equity.
\begin{itemize}
\item[(i)]  Let us assume that $X_0=k_0=0$, so that we have no equity and $K=k$, $Z=S^\Xb$. In this particular case: 
\begin{align*}
Y^*(\lambda_i)&=\frac{\Er_{\Qr^*}[(S^\Xb-K)^+\lambda_i]}{\Pc\((S^\Xb-K)^+\)}(K-S^\Xb)^++(S^\Xb\wedge K)\lambda_i\\
&=\alpha_i (K-S^{\Xb})^++(S^{\Xb}\wedge K)\frac{X_i}{S^\Xb}\\
&= \[X_i+\alpha_i\(k-S^{\Xb}\)\]\one_{\{S^{\Xb}\leq k\}}+X_i\(\frac{K}{S^{\Xb}}\wedge 1\)\one_{\{S^{\Xb}> k\}}
\end{align*}
that is exactly expression (\ref{formYi}) with $\alpha_i$ as in Theorem \ref{propFairStandard}. This proves that 
$$
Y_i=Y^*(\lambda_i)\in\Ar^\Xb(K).
$$
Then, the fact that $\Yb=(Y_0,...,Y_N)$ is a fair payoff for the risk $\Xb=(X_0,...,X_N)$ follows from the fact that $\boldsymbol \pi=\(0,\pi_1,...,\pi_N\) \in \Cc(c^\Xb)$ and the above. Indeed   (\ref{intermineq}) writes 
$$
c^{\Xb-\Yb}(S)+\sum_{i\in S}\pi_i \leq c^\Xb(S),\quad \forall S\in\Nc.
$$

\item[(ii)] Suppose $k_0>0$ that is,  $K=k+k_0>k$. The presence of equity  introduces an asymmetry between players, as equityholders have lower priority with no payment in case of default. In addition, the presence of equity shrinks the default event from $\{Z>K\}$ to $\{Z>K+k_0\}=\{S^\Xb>K\}$. In this situation, $Y^*(\lambda_i)\neq Y_i$ and even $Y^*(\lambda_i)\notin \Ar^\Xb(K)$. 
 In view of  (\ref{intermineq}), to prove fairness of $\Yb$ it is sufficient to show
\begin{equation}\label{ineqint}
c^{\Xb-\Yb}(S)\leq \Pc\(\sum_{i\in S}(X_i -Y^*(\lambda_i))\),\quad \forall S\in\Nc.
\end{equation}

We observe that for all $\Yb\in\Ar^\Xb(K)$ we have: 
\begin{equation*}
\{\exists i \in \Nc: X_i-Y_i>0\}\subset\cap_{i\in\Nc}\{X_j-Y_j\geq 0\}=\{S^{\Xb}\geq  k\}=\{Z\geq K\}
\end{equation*}
and  also
$$
\{\exists i\in\Nc: X_i-Y_i<0 \}\subset\cap_{i\in\Nc}\{X_j-Y_j\leq 0\}=\{S^{\Xb}\leq  k\} =\{Z\leq K\}.
$$
 Same expressions hold true if we replace $Y_i$ with $Y^*(\lambda_i)$. 
 
 Therefore, for any $\Yb\in\Ar^\Xb(K)$, $\Yb$ vector of standard payoffs, we have 
$$
\(\sum_{i\in S}(X_i-Y_i)\)^+=\sum_{i\in S}(X_i-Y_i)^+
$$ 
and 
$$
\(\sum_{i\in S}(X_i-Y_i)\)^-=\sum_{i\in S}(X_i-Y_i)^-=\sum_{i\in S}\alpha_i (K-Z)^+
$$
(the last equality comes from the definition of the standard payoffs).  

Let us consider $S\subset\Nc\setminus\{0\}$. Using the commonotonicity of $\Pc$ we have  
\begin{align*}
c^{\Xb-\Yb}(S)&=\Pc \(\sum_{i\in S}(X_i-Y_i)\) \\
&= \Pc\(\sum_{i\in S}(X_i-Y_i)^+\)+\Pc\(-\sum_{i\in S}(X_i-Y_i)^-\)\\
&= \Pc\(\sum_{i\in S}(X_i-Y_i)^+\)+\sum_{i\in S}\alpha_i \Pc\(-(K-Z)^+\) \\
&= \Pc\(\sum_{i\in S}(X_i-Y_i)^+\)-\sum_{i\in S}\alpha_i \Pc\((Z-K)^+\)\\
&= \Pc\(\sum_{i\in S}\frac{X_i}{S^\Xb}(S^\Xb-K)^+\)-\Er_{\Qr^*}\[\sum_{i\in S}\frac{X_i}{S^\Xb}(S^\Xb-K)^+\]\\
&=\Pc(\eta(S))-\Er_{\Qr^*}[\eta(S)].
\end{align*}
For simplicity, we have denoted above $\eta(S):=\sum_{i\in S}\frac{X_i}{S^\Xb}(S^\Xb-K)^+$. Also we have used the expressions for $\alpha_i$ that are given in Theorem \ref{propFairStandard}. Similarily
\begin{align*}
\Pc\(\sum_{i\in S}(X_i -Y^*(\lambda_i))\)&=\Pc\(\lambda(S)(Z-K)^+\)-\sum_{i\in S}\Er_{\alpha^*}[\lambda_i]\Pc\((Z-K)^+\)\\
&= \Pc\(\sum_{i\in S}\frac{X_i}{Z}(Z-K)^+\)-\Er_{\Qr^*}\[\sum_{i\in S}\frac{X_i}{Z}(Z-K)^+\]\\
&=\Pc(\xi(S))-\Er_{\Qr^*}[\xi(S)].
\end{align*}
We recall $\lambda_i=X_i/Z$. Also we denoted $\xi(S):= \sum_{i\in S}\frac{X_i}{Z}(Z-K)^+$.

 We notice that for any $S\subset\Nc\setminus \{0\}$ we have that the random variable  $\gamma(S): =\xi(S)-\eta(S)$ is commonotonic with $\eta(S)$. Indeed, 
 $$
 \gamma(S) =\sum_{i\in S}X_i\left\{\( \frac{K}{S^\Xb}-\frac{K}{S^\Xb+k_0}\)\one_{S^\Xb>K}+\(1-\frac{K}{S^\Xb+k_0}\)\one_{S^\Xb\in[k,K]}\right \}
 $$
 and 
 $$
 \eta(S)=\sum_{i\in S}X_i \(1-\frac{K}{S^\Xb}\)\one_{S^\Xb>K}.
 $$
Therefore:
\begin{align*}
\Pc\(\sum_{i\in S}(X_i -Y^*(\lambda_i))\)&=\Pc(\gamma(S))+\Pc(\eta(S))-\Er_{\Qr^*}[\gamma(S)+\eta(S)]= \\
&=\(\Pc(\gamma(S))-\Er_{\Qr^*}[\gamma(S)]\)+c^{\Xb-\Yb}(S)\geq c^{\Xb-\Yb}(S).\end{align*}
Thus we have proved (\ref{ineqint}) for all $S\subset\Nc\setminus \{0\}$ and it remains to prove the same holds when we include the element $\{0\}$.

We notice that $\pi_0=k_0=\Pc(X_0)$ and furthermore, for all  $S\subset\Nc\setminus \{0\}$, the random variable $X_0-Y_0$ is commonotonic with $\sum_{i\in S}(X_i-Y_i)$ which verifies that (\ref{ineqint})  holds for $S\cup\{0\}$ whenever it holds for  $S$.
\end{itemize}
\end{proof}

\section{Conclusion on the positive role of equity and bankruptcy procedures}
In this paper we used cooperative game theory in order to determine the fair part of agents in the surplus of a company. The analysis also provides an insight on the role of equity and the positive economic effects of the bankruptcy rules: by introducing a lower priority in default for shareholders as compared to the insured,  the impact of default of all players is reduced.  Inequality (\ref{ineqint}) is the mathematical transcription  of this fact. 
The reason behind is that in standard contracts, the shareholder's net position $X_0-Y_0$ is commonotonic with the total risk $S^\Xb$. But also, $X_0-Y_0$ is commonotonic with a net position $X_i-Y_i$ of  any insured agent. This property  is achieved due to the bankruptcy rules which prevent the shareholder to receive payments in default.  If shareholders were to have the same priority in default as insured agents, this would lead to a break in the commonotonicity property above mentioned and a destruction of value for all agents. Indeed, an increase of the residual risk of all agents in equilibrium occurs, as measured by the right hand term of (\ref{ineqint}).


\begin{thebibliography}{100}

\bibitem{ADEH1} Artzner, Ph., F. Delbaen, J.-M. Eber, and D. Heath: {\em Thinking Coherently},  RISK, November 97, 68--71, (1997).

\bibitem{ADEH2} Artzner, Ph., F. Delbaen, J.-M. Eber, and D. Heath:  {\em Coherent Risk Measures}, Mathematical Finance {\bf 9}, 145--175, (1999).

\bibitem{AO} Artzner, P. and Ostroy: {\em Gradients, subgradients and economic equilibria}, Adv. in Appl. Math., {\bf 4},  pp. 245--259 (1984).

\bibitem{Au} Aubin,, J.P. : {\em C\oe urs et \'equilibres des jeux flous \`a paiements lat\'eraux}, C.R.Acad.Sci. Paris S A, 279, pp. 891--894, (1974). 

\bibitem{BarrElK05}  Barrieu, P. and El Karoui, N.: {\em Inf-convolution of risk measures and optimal risk transfer. } Finance and Stochastics 9, 269--298 (2005).

\bibitem{BarrElK05a} Barrieu, P. and  El Karoui, N.: {\em Pricing, hedging and optimally designing derivatives via minimization of risk measures.} R. Carmona (eds.): Volume on Indifference Pricing. Princeton University Press (2005).

\bibitem{BH} Billera, L.J. and Heath, C.C.:  {\em Allocation of costs: a set of axioms yielding a unique procedure}, Mathematics of Operations Research 1, pp. 32--39, (1982).

 \bibitem{BurRus06} Burgert, C. and R\"uschendorf, L.: {\em On the optimal risk allocation problem.} Statistics \& Decisions 24, 153--171 (2006).
 
\bibitem{Borch62}  Borch, K.: {\em Equilibrium in a reinsurance market.} Econometrica {\bf 30}, 424--444 (1962).

\bibitem{D1} Delbaen, F.:  {\em Convex Games and Extreme Points},  Journ. Math.
Anal. Appli., {\bf 45}, pp. 210--233, (1974).

\bibitem{Pisa} Delbaen, F:   {\em Coherent Risk Measures}, Lectures given at the
Cattedra Galileiana at the Scuola Normale Superiore di Pisa, March 2000,   Published by the
{\em Scuola Normale Superiore di Pisa}, (2002).

\bibitem{Del} Delbaen, F.: {\em Coherent Risk Measures on General Probability Spaces} in Advances in Finance and Stochastics, pp.  1--37, Springer, Berlin (2002).


\bibitem{FDbook} Delbaen, F:   {\em Monetary Utility Functions}, Lectures held in 2008 and published in the series ``Lecture Notes of the University of Osaka",    (2011).

\bibitem{Denn1} Denneberg, D.: Verzerrte Wahrscheinlichkeiten in der
Versicherungsmathematik, quantilabh\"angige Pr\"amienprinzipen. {\em 
Mathematik-Arbeitspapiere}, {\bf 34}, Universit\"at Bremen, (1989).

\bibitem{DG} Deprez, O. and  Gerber, H.U:   {\em On convex principles of premium  calculation.}  Insurance: Mathematics and Economics, pp. 179--189, (1985).

\bibitem{HeaKu04} Heath, D. and Ku, H.: {\em  Pareto equilibria with coherent measures of risk.} Mathematical Finance 14, 163-172 (2004).

\bibitem{FilKup08} Filipovi\'c, D. and Kupper, M.:  {\em  Equilibrium Prices For Monetary Utility Functions}, International Journal of Theoretical and Applied Finance (IJTAF), World Scientific Publishing Co. Pte. Ltd., vol. {\bf 11}(03), pages 325--343 (2008).

\bibitem{JouSchTou08}  Jouini, E., Schachermayer, W., and Touzi, N.: {\em Optimal Risk Sharing for Law Invariant Monetary Utility Functions.} Mathematical Finance, 18, 269--292. (2008).

\bibitem{Schm1} Schmeidler, D.:   {\em Cores of Convex Games.}  J. Math. Anal. Appli., {\bf 40},  pp. 214--225, (1972).


\bibitem{Schm} Schmeidler, D.: {\em  Integral Representation without Additivity}, 
Proc. Amer. Math. Soc., {\bf 97}, 255--261,  (1986).

\bibitem{Shapley} Shapley, L.:  {\em Notes on N-person games, VII: Cores of convex games. }Rand Memorandum, RM 4571 PR, Rand Corp, Santa Monica, Calif., (1965).

\bibitem{Yaa} Yaari, M. E.: {\em The dual theory of choice under risk}. Econometrica,
{\bf 55}, 95--115, (1987).


\end{thebibliography}
\end{document}